\documentclass{paper}
\usepackage{amsmath,amsthm,mleftright}
\usepackage{amssymb,mathtools}
\DeclareMathAlphabet\mathcal{OMS}{cmsy}{m}{n}
\DeclareMathAlphabet\mathbfcal{OMS}{cmsy}{b}{n}
\usepackage{tikz-cd}
\usepackage{xspace}
\usepackage{makecell}
\usepackage{graphicx,booktabs}
\usepackage{ragged2e}
\usepackage{booktabs,multirow,tabularx}
\usepackage{paralist}
\usepackage{hyperref}
\hypersetup{%
    pdfmenubar=true,       
    pdffitwindow=false,     
    pdfstartview={FitH},    
    pdftitle={Image Reconstruction for Multispectral CT},
    colorlinks=true,       
    linkcolor=red,          
    citecolor=red,        
    filecolor=magenta,      
    urlcolor=cyan,           
    pdfborder = {0,0,0}
}
\usepackage{acro}
\usepackage{cleveref}
\usepackage{algorithm}
\usepackage[noend]{algpseudocode}
\makeatletter
\def\BState{\State\hskip-\ALG@thistlm}
\makeatother
\Crefname{equation}{}{}
\crefname{equation}{}{}
\usepackage[title]{appendix}
\usepackage{todonotes}
\usepackage{inputenc}
\theoremstyle{plain}
 \newtheorem{theorem}{Theorem}
  
 \newtheorem{lemma}{Lemma}

\theoremstyle{definition}

\theoremstyle{remark}

\usepackage{bm}
\usepackage{adjustbox}
\newcommand{\trans}{^{\mathsf{T}}}
\newcommand{\diff}{\mathrm{d}}

\begin{document}

\title{Efficient and Accurate Image Reconstruction for Geometric-Inconsistent Multispectral CT with Ray-Dependent Energy Spectra}
\author{Ziqiang Zhang\thanks{State Key Laboratory of Mathematical Sciences, Academy of Mathematics and Systems Science, Chinese Academy of Sciences, Beijing 100190, China; University of Chinese Academy of Sciences, Beijing 100190, China.}, 
Chong Chen\thanks{State Key Laboratory of Mathematical Sciences, Academy of Mathematics and Systems Science, Chinese Academy of Sciences, Beijing 100190, China.}
}

\maketitle

\begin{abstract}
In practical multispectral computed tomography (MSCT), the scanning geometric parameters under different X-ray energy spectra are often inconsistent, and the distributions of the energy spectra are even ray-dependent. However, existing algorithms cannot effectively and accurately solve the associated image reconstruction problem. To address this limitation, using the proposed aggregated energy spectra, we approximate the Jacobian matrix of the nonlinear forward operator at certain special points (e.g., the zero point) as a block product of a diagonal matrix composed of projection matrices and a very small-scale matrix, and then based on this matrix with a special structure, propose an efficient and accurate image reconstruction algorithm tailored for geometric-inconsistent MSCT with ray-dependent energy spectra. Under appropriate conditions, we establish the convergence theory for the proposed algorithm. Furthermore, numerical experiments using both noiseless and noisy projection data are conducted to verify the performance of the proposed algorithm, which demonstrate that the efficiency and accuracy of this algorithm are much higher than existing algorithms, offering the flexibility and scalability to accommodate various MSCT imaging configurations.  
\end{abstract}


\section{Introduction}
	
	Multispectral computed tomography (MSCT) is an advanced imaging modality that uses different X-ray energy spectra or different energy bins of the spectrum to scan the object, and reconstructs basis images from projection data obtained by nonlinear imaging processes. The resulting basis images can then be linearly combined to form virtual monochromatic images (VMIs) or obtain the spatial distribution of linear attenuation coefficient (LAC) at any X-ray energy of interest \cite{Alvarez_1976, Hounsfield1973ComputerizedTA, McCollough_2015}. Unlike conventional CT \cite{shepp1978computerized}, MSCT takes photon energy into account, thereby providing significantly better image quality under the same radiation dose \cite{Alvarez_1979}. Owing to its advantages in artifact suppression \cite{Guggenberger2012, McCollough_2015}, material distinguishability \cite{Patino_2016} and quantitative analysis \cite{Fletcher_2024}, MSCT has become an active research topic in the field of imaging \cite{GREFFIER2023167}, and is widely applied in medicine \cite{Sahbaee_2026,jcm10245757,Treb2025TechnicalPO} and industry \cite{XU201899, LEE2018105}.
	
	A common issue in practical MSCT is geometric-inconsistent scanning, which implies that X-rays from different energy spectra do not share identical scanning geometric parameters. In other words, there exists at least one X-ray path that is not intersected by X-rays from certain energy spectra \cite{maass2011exact}. This situation inevitably occurs to varying degrees in some existing practical systems, such as sequential dual-energy CT, dual-source dual-detector systems, and fast kVp switching \cite{McCollough_2015}.
	
	In addition to the geometric inconsistency, the X-ray energy spectrum itself depends inherently on the ray path determined by the scan geometry parameters \cite{Chen_2017}, and consequently, different X-ray beams most likely exhibit distinct spectral distributions. This dependence stems primarily from physical hardware configurations, such as the presence of a bow-tie filter near the X-ray source, as well as acquisition schemes where multiple measurements are performed along a single ray, for instance, using multiple energy bins in a photon-counting detector. Moreover, the effective energy spectrum of a CT system is jointly determined by the intrinsic X-ray tube spectrum and the energy-dependent response of the detector, further reinforcing its ray-dependence \cite{Ren_2021}.

	Traditional algorithms for MSCT image reconstruction fall into three categories: image-domain, projection-domain, and one-step methods. Image-domain methods \cite{image_2009,Niu_2014} first reconstruct images independently from the measured polychromatic projections and then linearly combine them to form basis images. However, such linear combination fails to capture the nonlinear relationship between the basis images and the polychromatic projections, which inevitably introduces beam-hardening artifacts \cite{TIBKAT:1810057744}. Projection-domain methods \cite{Alvarez_1976,MACOVSKI1976325,Bal_2022} first decompose the polychromatic projections into basis sinograms, and then invert them individually to obtain the basis images. However, these methods require that scanning geometric parameters remain consistent across different energy spectra. As a result, they cannot be directly applied to geometric-inconsistent MSCT. Moreover, both image-domain and projection-domain methods can be classified as the category of two-step methods, as they split the reconstruction workflow into two sequential phases. Therefore, such two-step methods suffer from a fundamental drawback: irreversible information loss occurring in the first step cannot be compensated in the second, which ultimately compromises their accuracy in practice \cite{Mory_2018}. 
	
	Given the limited accuracy of these two-step methods, research focus has gradually shifted toward one-step methods which can be seen as iterative methods \cite{Barber_2016, Cai2013, 993128, 6805661, 7979598, 6601625}. Different from two-step methods, the one-step method directly maps polychromatic projections to basis images during each iteration. Zhao et al. \cite{Zhao_2015} extended the algebraic reconstruction technique (ART) for conventional CT image reconstruction to dual-energy CT (DECT) and named it extended ART (E-ART). Before that, Natterer et al. referred to this method as ART for nonlinear problems \cite[section 5.7]{natterer2001}. Essentially, this method is an application of the nonlinear Kaczmarz method (NKM) for solving the corresponding nonlinear system of equations in DECT image reconstruction \cite{Gao_NKM}.  Chen et al. \cite{CHEN2021101821} proposed a nonlinear least squares model with total variation and nonnegativity constraints, and designed a nonconvex primal-dual (NCPD) method. Gao et al. \cite{Gao_2022} further formulated this image reconstruction problem as a class of nonconvex nonsmooth optimization problems, and developed an extended primal-dual algorithmic framework. The aforementioned one-step methods can reconstruct high-quality basis images for the situations where both of geometric inconsistency and ray dependence coexist in MSCT. However, they typically suffer from slow convergence speeds \cite{Mory_2018}.
	
	More recently, Gao et al. \cite{AFIRE2025} introduced an algorithm called AFIRE by leveraging the simplified Newton method. Although the algorithm AFIRE enables efficient and accurate reconstruction of basis images in geometric-inconsistent MSCT under mildly full scan conditions, its main limitation lies in the underlying assumption that the distributions of energy spectra are ray-independent.
	
	Existing methods struggle to solve the image reconstruction problem efficiently and accurately for geometric-inconsistent MSCT with ray-dependent energy spectra. To address this limitation, we propose a novel reconstruction method tailored to this scenario. The core idea of this method is to decouple the weight coefficients induced by the ray dependence that involved in the Jacobian matrix of the nonlinear forward operator by selecting special points (e.g., the zero point) and employing the proposed aggregated energy spectra. Subsequently, an approximation of the Jacobian matrix is yielded at those special points, which takes the form of a block product consisting of a diagonal matrix of projection matrices and a very small-scale matrix. Inspired by the algorithm AFIRE, the iterative scheme is finally designed by using the inverse of this approximate Jacobian matrix, enabling both efficient and accurate image reconstruction.
	
	The rest of this paper is organized as follows. We introduce the considered forward data model for MSCT in \cref{section2}. \Cref{section3} presents the proposed reconstruction algorithm. In \cref{section4}, we provide the convergence analysis of the proposed algorithm. \Cref{section5} reports some numerical results to validate the performance of the proposed algorithm. Finally, we conclude the paper in \cref{section6}.

\section{The considered forward data model}\label{section2}
	This section briefly introduces the forward data model in geometric-inconsistent MSCT with ray-dependent energy spectra. In MSCT, the image reconstruction problem aims to determine the function of LAC $\mu : \Omega_1 \times \Omega_2 \rightarrow \mathbb{R}$ from transmission measurements acquired with multiple energy spectra, where $\Omega_1\subset \mathbb{R}^{n}$ ($n=2$ or $3$) and $\Omega_2 \subset \mathbb{R}$ denote spatial domain and energy domain, respectively. The LAC function $\mu$ can be typically decomposed into a linear combination of multiple basis images \cite{Alvarez_1976, Hsieh2022CT} as follows
	\begin{equation}\label{mdcc1}
		\mu (\bm{y}, E)=\sum_{d=1}^{D} b_{d}(E) f_{d}(\bm{y}),
	\end{equation}
	where $b_{d}$ is the function of mass attenuation coefficient (MAC) of the $d$-th basis material (e.g., water and bone), $f_d$ represents the material-equivalent density function of the $d$-th basis material, and $D$ denotes the number of basis materials. Since the MAC is generally known a priori, determining $\mu$ is equivalent to reconstructing the $D$ basis images.
	
For given spectrum $q$, data is measured along the X-ray specified by $L(\theta^{[q]})$, where $\theta^{[q]}$ is the parameter determining the ray position. According to Beer--Lambert law \cite{Alvarez_1976, natterer2001}, the measured data can be formulated as, for $q = 1, \ldots, Q$, 
	\begin{equation}\label{dmcc1}
		g_{L(\theta^{[q]})}^{[q]} = \ln\int_{0}^{E_{\text{max}}}s_{L(\theta^{[q]})}^{[q]}(E)\exp\left(-\int_{L(\theta^{[q]})} \mu(\bm{y}, E) \mathrm{d} l\right)\mathrm{d} E,
	\end{equation}
	where $Q$ is the number of energy spectra, $E_{\max}$ denotes the maximum energy in the scan, and $s_{L(\theta^{[q]})}^{[q]}$ is the function of energy spectrum for ray $L(\theta^{[q]})$ under spectrum $q$ such that  
	\[
		s_{L(\theta^{[q]})}^{[q]} \geq 0 \quad  \text{and}\quad \int_{0}^{E_{\text{max}}}s_{L(\theta^{[q]})}^{[q]}(E) \mathrm{d} E = 1.
	\]
Substituting \eqref{mdcc1} into \eqref{dmcc1} yields the continuous-to-continuous (CC)-data model
	\begin{equation}\label{dmcc2}
	g_{L(\theta^{[q]})}^{[q]} = \ln\int_{0}^{E_{\text{max}}}s_{L(\theta^{[q]})}^{[q]}(E)\exp\left(-\sum_{d=1}^D b_d(E)\int_{L(\theta^{[q]})}f_d(\bm{y}) \mathrm{d} l\right) \mathrm{d} E. 
	\end{equation}

To align with practical data acquisition and image reconstruction which are inherently discrete in energy, rays, and spatial domain, we obtain a discrete-to-discrete (DD)-data model based on the CC-data model \eqref{dmcc2}. Assume that the energy interval is uniformly divided into $M$ subintervals of size $\Delta_{E}$, and the data are detected for discrete directions $\theta_j^{[q]}$ with $j=1,\dots, J^{[q]}$ under the $q$-th spectrum. Each basis image function $f_d$ is discretized by an $I$-dimensional vector $\mathbf{f}_d=\left[f_{d1},\dots,f_{dI}\right]\trans$, where $f_{di}$ denotes the gray value of the $d$-th discrete basis image at pixel/voxel $i$. For convenience, we assemble all the basis images into a combined vector $\mathbf{f}=\left[\mathbf{f}_1\trans,\dots,\mathbf{f}_D\trans\right]\trans$. Then the forward operator of ray $j$ under spectrum $q$ can be modeled as
	\begin{equation}\label{dm2}
		K_j^{[q]}(\mathbf{f})=\ln \sum_{m=1}^{M} s_{jm}^{[q]} \exp\left(-\sum_{d=1}^{D} b_{dm}\sum_{i=1}^{I} p_{ji}^{[q]}f_{di}\right),
	\end{equation}
	where $s_{jm}^{[q]}=s_{j}^{[q]}(m \Delta_{E}) \Delta_{E}$ satisfying
	\begin{equation}\label{eq:s_condition}
		s_{jm}^{[q]} \geq 0  \quad  \text{and}\quad \sum_{m=1}^{M} s_{jm}^{[q]}=1,
	\end{equation}
	$b_{dm}=b_d(m \Delta_{E})$, and $p_{ji}^{[q]}$ is often taken as the intersection length of ray $j$ with pixel/voxel $i$ \cite{Chen03072022}. 
	
Let $\mathbf{s}_{j}^{[q]}=[s_{j1}^{[q]},\dots,s_{jM}^{[q]}]\trans$ denote the ray-specific energy spectrum vector for ray $j$ under spectrum $q$. If $\mathbf{s}_{1}^{[q]}=\dots=\mathbf{s}_{J^{[q]}}^{[q]}$, the spectrum is termed ray-independent; otherwise, it is referred to as ray-dependent. In realistic scenarios, spectra are often ray-dependent due to factors such as bow-tie filters and geometric parameter variations. The estimation and correction of such ray-dependent spectra can be referred to \cite{10462530,7906625}. Moreover, define $\mathbf{b}_d=[b_{d1},\dots,b_{dM}]\trans$ as the MAC vector of the $d$-th basis material, and $\mathbf{p}^{[q]}_j=[p_{j1}^{[q]},...,p_{jI}^{[q]}]\trans$ as the projection vector for ray $j$ under spectrum $q$. The corresponding projection matrix under spectrum $q$ is then expressed as $\mathbf{P}^{[q]} = [\mathbf{p}_1^{[q]}, \dots, \mathbf{p}_{J^{[q]}}^{[q]}]\trans$.  Note that $\mathbf{P}^{[1]}=\dots=\mathbf{P}^{[Q]}$ holds if and only if the scanning configuration is geometric-consistent; otherwise, it corresponds to the geometric-inconsistent case.
	
Hence, the MSCT image reconstruction problem amounts to determining the combined basis image vector $\mathbf{f}$ from the measured projection data $\{g_j^{[q]}\}_{q,j=1}^{Q,J^{[q]}}$ by solving the following system of nonlinear equations
	\begin{equation}\label{nf1}
		K_j^{[q]}(\mathbf{f})=g_j^{[q]} \quad\text{for}\ q=1,\dots,Q\ \text{and}\ j=1,\dots,J^{[q]}.
	\end{equation}
	The compact form of \eqref{nf1} can be expressed by
	\begin{equation}\label{nf2}
		\mathbf{K}(\mathbf{f})=\mathbf{g},
	\end{equation}
	where
	\[
		\mathbf{K}(\mathbf{f})=\begin{bmatrix}
		\mathbf{K}^{[1]}(\mathbf{f})\\ \vdots \\ \mathbf{K}^{[Q]}(\mathbf{f}) 
		\end{bmatrix}, \quad
		\mathbf{K}^{[q]}(\mathbf{f})=\begin{bmatrix}
			 K_1^{[q]}(\mathbf{f}) \\ \vdots \\ K_{J^{[q]}}^{[q]}(\mathbf{f})
		\end{bmatrix}, \quad
		\mathbf{g}=\begin{bmatrix}
			\mathbf{g}^{[1]} \\ \vdots \\ \mathbf{g}^{[Q]} 
		\end{bmatrix}, \quad
		\mathbf{g}^{[q]}=\begin{bmatrix}
			g_1^{[q]}\\ \vdots \\ g_{J^{[q]}}^{[q]} 
		\end{bmatrix}.
	\]

	\section{The proposed algorithm} \label{section3}
	
	In this section, we will propose the efficient and accurate image reconstruction algorithm for  geometric-inconsistent MSCT with ray-dependent energy spectra. Following the AFIRE algorithm, the iterative scheme for solving the nonlinear system \eqref{nf2} is expressed as
	\begin{equation}\label{sni1}
		\mathbf{f}^{k+1}=\mathbf{f}^k+\bigl(D\mathbf{K}(\hat{\mathbf{f}}) \bigr)^{-1}\left(\mathbf{g}-\mathbf{K}(\mathbf{f}^k)\right),
	\end{equation}
	where $\mathbf{f}^k$ denotes the combined basis image vector at the $k$-th iteration, and $\hat{\mathbf{f}}$ is a fixed combined basis image vector, $D\mathbf{K}(\hat{\mathbf{f}})$ denotes the Jacobian matrix of the forward operator $\mathbf{K}$ evaluated at $\hat{\mathbf{f}}$, constructed by vertically stacking $Q$ Jacobian submatrices as follows
	\[
		D \mathbf{K}(\hat{\mathbf{f}})=\begin{bmatrix}
			D \mathbf{K}^{[1]}(\hat{\mathbf{f}}) \\ \vdots \\ D \mathbf{K}^{[Q]}(\hat{\mathbf{f}})
		\end{bmatrix}.
	\]
Here the explicit form of the Jacobian submatrix $D \mathbf{K}^{[q]}(\hat{\mathbf{f}})$ is derived as  
	\[
		D \mathbf{K}^{[q]}(\hat{\mathbf{f}})=\begin{bmatrix}
			w_{11}^{[q]}(\hat{\mathbf{f}})(\mathbf{p}_{1}^{[q]})\trans & \cdots & w_{D1}^{[q]}(\hat{\mathbf{f}})(\mathbf{p}_{1}^{[q]})\trans \\
			\vdots                           & \ddots & \vdots                           \\
			w_{1J^{[q]}}^{[q]}(\hat{\mathbf{f}})(\mathbf{p}_{J^{[q]}}^{[q]})\trans & \cdots & w_{DJ^{[q]}}^{[q]}(\hat{\mathbf{f}})(\mathbf{p}_{J^{[q]}}^{[q]})\trans
		\end{bmatrix},
	\]
	where the involved weight coefficient $w_{dj}^{[q]}(\hat{\mathbf{f}})$ is computed by
	\[
		w_{dj}^{[q]}(\hat{\mathbf{f}})=-\frac{\displaystyle\sum_{m=1}^{M} b_{dm}s^{[q]}_{jm} \exp\left(-\sum_{d^{\prime}=1}^{D} b_{d^{\prime}m} (\mathbf{p}_{j}^{[q]})\trans \hat{\mathbf{f}}_{d^{\prime}} \right) }{\displaystyle\sum_{m=1}^{M} s^{[q]}_{jm} \exp\left(-\sum_{d^{\prime}=1}^{D} b_{d^{\prime}m} (\mathbf{p}_{j}^{[q]})\trans \hat{\mathbf{f}}_{d^{\prime}} \right)}.
	\]
	
Due to the complex structure of the Jacobian matrix, its direct inversion is computationally prohibitive, which makes \eqref{sni1} impractical to implement. Nevertheless, we observe that eliminating the dependence of the weight coefficients on the ray indices can endow the Jacobian matrix with a special structure that enables efficient inversion.
	
Specifically, the dependence of the weight coefficients $w_{dj}^{[q]}(\hat{\mathbf{f}})$ on the ray index $j$ originates from two sources: the ray-dependent energy spectrum $\mathbf{s}_{j}^{[q]}$ and the basis sinogram element $(\mathbf{p}_{j}^{[q]})\trans \hat{\mathbf{f}}_{d^{\prime}}$. 
	
To decouple the basis sinogram from the ray index, we can fix the Jacobian matrix at some special points, analogous to the AFIRE algorithm, for which the corresponding basis sinograms are required to be constant vectors. Specifically, for such a special point $\hat{\mathbf{f}}$, there exists a constant set $\{C_{d}^{[q]}\}_{q,d=1}^{Q,D}$ such that the following identity holds
	\begin{equation}\label{condition1}
		(\mathbf{p}_{j}^{[q]})\trans \hat{\mathbf{f}}_d \equiv C_{d}^{[q]}, \quad \forall j\in\{1,\dots,J^{[q]}\}.
	\end{equation}
	Then the weight coefficients $w_{dj}^{[q]}(\hat{\mathbf{f}})$ are translated into 
	\begin{equation}\label{wc2}
		w_{dj}^{[q]}(\hat{\mathbf{f}})=-\frac{\displaystyle\sum_{m=1}^{M} b_{dm}s^{[q]}_{jm} \exp\left(-\sum_{d^{\prime}=1}^{D} b_{d^{\prime}m} C_{d^{\prime}}^{[q]} \right) }{\displaystyle\sum_{m=1}^{M} s^{[q]}_{jm} \exp\left(-\sum_{d^{\prime}=1}^{D} b_{d^{\prime}m} C_{d^{\prime}}^{[q]} \right)}.
	\end{equation}
	Note that such a special point always exists. For instance, taking $\hat{\mathbf{f}}=\mathbf{0}$ satisfies condition \eqref{condition1} with $C_d^{[q]}=0$ for all $q$, $d$. Substituting $C_d^{[q]}=0$ into \eqref{wc2}, and by \eqref{eq:s_condition}, the corresponding weight coefficients become
	\[
		w_{dj}^{[q]}(\mathbf{0}) = -\sum_{m=1}^{M} b_{dm}s^{[q]}_{jm}.
	\]
	
To further decouple the energy spectrum from the ray index, we employ the aggregated energy spectrum $\bar{\mathbf{s}}^{[q]}$ in place of the ray-dependent energy spectra $\{\mathbf{s}_{j}^{[q]}\}_{j=1}^{J^{[q]}}$. The aggregated spectrum can be chosen, for example, as the mean or median of all or a subset of the ray-dependent energy spectra. For instance, the mean aggregation is given by
	\begin{equation}\label{as1}
		\bar{\mathbf{s}}^{[q]} = [\bar{s}_1^{[q]}, \dots, \bar{s}_M^{[q]}]\trans , \quad \bar{s}_m^{[q]}=\frac{1}{J^{[q]}}\sum_{j=1}^{J^{[q]}} s_{jm}^{[q]}.
	\end{equation}
	By substituting the ray-dependent spectrum in \eqref{wc2} with the aggregated spectrum from \eqref{as1}, we obtain an approximate, ray-independent form of the weight coefficient, denoted as $\phi_{d}^{[q]}$, whose explicit expression is
	\begin{equation}\label{wc3}
		\phi_{d}^{[q]}(\hat{\mathbf{f}})=-\frac{\displaystyle\sum_{m=1}^{M} b_{dm}\bar{s}_m^{[q]} \exp\left(-\sum_{d^{\prime}=1}^{D} b_{d^{\prime}m} C_{d^{\prime}}^{[q]} \right) }{\displaystyle\sum_{m=1}^{M} \bar{s}_m^{[q]} \exp\left(-\sum_{d^{\prime}=1}^{D} b_{d^{\prime}m} C_{d^{\prime}}^{[q]} \right)}.
	\end{equation}
	
Through the above decoupling, the Jacobian submatrix evaluated at the special point $\hat{\mathbf{f}}$ can be approximated as follows
\[
			D \mathbf{K}^{[q]}(\hat{\mathbf{f}})  \approx  \begin{bmatrix}
				\phi_{1}^{[q]}(\hat{\mathbf{f}}) \mathbf{P}^{[q]} & \cdots & \phi_{D}^{[q]}(\hat{\mathbf{f}}) \mathbf{P}^{[q]}
			\end{bmatrix}.
	\]
Accordingly, the Jacobian matrix at $\hat{\mathbf{f}}$ can be approximated as 
	\begin{equation}\label{ajm}
		D \mathbf{K}(\hat{\mathbf{f}}) \approx \mathbf{J}(\hat{\mathbf{f}}) = \begin{bmatrix}
			\phi_{1}^{[1]}(\hat{\mathbf{f}}) \mathbf{P}^{[1]} & \cdots & \phi_{D}^{[1]}(\hat{\mathbf{f}}) \mathbf{P}^{[1]} \\
			\vdots & \ddots & \vdots \\
			\phi_{1}^{[Q]}(\hat{\mathbf{f}}) \mathbf{P}^{[Q]} & \cdots & \phi_{D}^{[Q]}(\hat{\mathbf{f}}) \mathbf{P}^{[Q]}
		\end{bmatrix} = \mathbf{P} \otimes \bm{\phi},
	\end{equation}
	where $\mathbf{J}(\hat{\mathbf{f}})$ denotes the approximate Jacobian matrix,  $\mathbf{P}=\texttt{diag}(\mathbf{P}^{[1]},\cdots,\mathbf{P}^{[Q]})$ is a block diagonal matrix consisting of the $Q$ projection matrices, $\bm{\phi}=(\phi_d^{[q]}(\hat{\mathbf{f}}))_{Q\times D}$ is a very small-scale matrix depending on $\hat{\mathbf{f}}$, and the operation $\otimes$ represents block multiplication between two matrices. For simplicity, we drop the explicit dependence of $\phi_d^{[q]}$ on $\hat{\mathbf{f}}$ throughout the rest of the paper.
	
	Note that when $Q = D$, and both $\bm{\phi}$ and $\mathbf{P}^{[q]}$ for all $q$ are invertible, the approximate Jacobian matrix $\mathbf{J}(\hat{\mathbf{f}})$ becomes invertible. Indeed, according to \eqref{ajm}, its inverse can be readily obtained as $$\bigl(\mathbf{J}(\hat{\mathbf{f}})\bigr)^{-1}=\bm{\phi}^{-1} \otimes \mathbf{P}^{-1},$$ where $\mathbf{P}^{-1}=\texttt{diag}\bigl( (\mathbf{P}^{[1]})^{-1}, \cdots, (\mathbf{P}^{[Q]})^{-1}  \bigr)$. Substituting the approximate Jacobian matrix $\mathbf{J}(\hat{\mathbf{f}})$ for the exact Jacobian matrix $D \mathbf{K}(\hat{\mathbf{f}})$, the iterative scheme \eqref{sni1} becomes
	\begin{equation}\label{sni2}
		\mathbf{f}^{k+1}
		=\mathbf{f}^k+\bm{\phi}^{-1} \otimes \mathbf{P}^{-1} (\mathbf{g}-\mathbf{K}(\mathbf{f}^k)).
	\end{equation}
	
In practice, the number of energy spectra $Q$ does not have to equal the number of basis materials $D$, and it only needs to satisfy $Q\geq D$. Hence, the matrix $\bm{\phi}$ is not necessarily square, and its pseudoinverse $\bm{\phi}^{\dagger}$, for instance, $\bm{\phi}^{\dagger} = (\bm{\phi}\trans\bm{\phi})^{-1}\bm{\phi}\trans$, can be used in place of $\bm{\phi}^{-1}$. Meanwhile, the exact inversion of the matrix $\mathbf{P}$ is infeasible in practice. In real MSCT systems, the projection matrices $\mathbf{P}^{[1]},\dots,\mathbf{P}^{[Q]}$ are not necessarily invertible square matrices, and even when they are, the computational cost of large-scale matrix inversion is prohibitive. We thus adopt an approximation for the inverse of $\mathbf{P}$, and the resulting iterative scheme corresponding to \eqref{sni2} becomes
	\begin{equation}\label{sni3}
		\mathbf{f}^{k+1} = \mathbf{f}^k + \bm{\phi}^{\dagger} \otimes \mathbf{P}^{+} \left(\mathbf{g}-\mathbf{K}(\mathbf{f}^k)\right),
	\end{equation}
	where $\mathbf{P}^{+}=\texttt{diag}\bigl( (\mathbf{P}^{[1]})^{+}, \cdots, (\mathbf{P}^{[Q]})^{+} \bigr)$ represents the approximate inverse of $\mathbf{P}$, and $(\mathbf{P}^{[q]})^{+}$ denotes the approximate inverse of $\mathbf{P}^{[q]}$. The implementation of $(\mathbf{P}^{[q]})^{+}$ can be tailored according to the actual scanning geometry. Specifically, for two-dimensional scenarios, one may employ methods such as filtered back-projection (FBP), ART, optimization-based reconstruction, or deep learning-based reconstruction. For three-dimensional cone-beam imaging, the Feldkamp–Davis–Kress (FDK) algorithm can be adopted. Based on \eqref{sni3}, we summarize the proposed algorithm for solving \eqref{nf2} in \cref{alg:msct_recon}.
	
	\begin{algorithm}[htb]
		\caption{The proposed algorithm for solving \eqref{nf2}}
		\label{alg:msct_recon}
		\begin{algorithmic}[1]
			\State \emph{Initialize}: Given $Q\geq D$, the measured data $\mathbf{g}$, the ray-dependent energy spectra $\{\mathbf{s}_{j}^{[q]}\}_{q,j=1}^{Q,J^{[q]}}$, the MACs $\{\mathbf{b}_d\}_{d=1}^{D}$, the geometric-inconsistent projection matrices $\{\mathbf{P}^{[q]}\}_{q=1}^{Q}$, a certain constant set $\{C_d^{[q]}\}_{q,d=1}^{Q,D}$, a maximum iteration number $K$, 
			and a tolerance $\varepsilon > 0$. Initialized the basis images $\mathbf{f}^0$. Let $k \leftarrow 0$.
			
			\State Compute the matrix $\bm{\phi}=(\phi_d^{[q]})_{Q\times D}$ from \eqref{wc3} and its pseudoinverse $\bm{\phi}^{\dagger}$.
			
			\State \emph{Loop}:
			\State \quad Compute the residual 
			\begin{equation*}
				\begin{bmatrix}
					\mathbf{r}^{[1],k} \\ \vdots \\ \mathbf{r}^{[Q],k} 
				\end{bmatrix}
				\leftarrow 
				\begin{bmatrix}
					\mathbf{g}^{[1]} \\ \vdots \\ \mathbf{g}^{[Q]} 
				\end{bmatrix}
				- \begin{bmatrix}
					\mathbf{K}^{[1]}(\mathbf{f}^k) \\ \vdots \\ \mathbf{K}^{[Q]}(\mathbf{f}^k) 
				\end{bmatrix}.
			\end{equation*}
			
			\State \quad Update $\mathbf{f}^{k+1}$ by
			\begin{equation*}
				\begin{bmatrix}
					\mathbf{f}^{k+1}_1 \\ \vdots \\ \mathbf{f}^{k+1}_D
				\end{bmatrix}
				\leftarrow 
				\begin{bmatrix}
					\mathbf{f}^{k}_1 \\ \vdots \\ \mathbf{f}^{k}_D
				\end{bmatrix}
				+ \bm{\phi}^{\dagger} \otimes
				\begin{bmatrix}
					(\mathbf{P}^{[1]})^{+}\mathbf{r}^{[1],k} \\ \vdots \\ (\mathbf{P}^{[Q]})^{+}\mathbf{r}^{[Q],k} 
				\end{bmatrix}.
			\end{equation*}
			
			\State \quad If $k \geq K$ or $\sum_{q=1}^{Q} \|\mathbf{r}^{[q],k}\| < \varepsilon $, then \textbf{output} $\mathbf{f}^{k+1}$; \newline \indent \hspace{-3.5mm}otherwise, let $k \gets k+1$, \textbf{goto} \emph{Loop}. 
			
		\end{algorithmic}
	\end{algorithm}

\section{Convergence theory}\label{section4}
	Here we will analyze the convergence of the proposed algorithm. Let $\mathbf{f}^{\ast}$ denote the solution to the MSCT image reconstruction problem \eqref{nf2}, provided it exists. Let $\|\cdot\|_F$ and $\sigma_{\min}(\cdot)$ be the Frobenius norm and the minimum singular value of a matrix, respectively. For $\sigma_{\min}(\cdot)\neq 0$, define the scaled condition number as $\kappa_F(\cdot)=\|\cdot\|_F/\sigma_{\min}(\cdot)$. In addition, let $\|\cdot\|_2$ denote the Euclidean norm of a vector or the spectral norm of a matrix.
	
\begin{lemma}\label{lemma:weight_diff}
For any $\delta\in(0,1)$, there exists a constant $\epsilon=2\delta/ (1-\delta) > 0$ such that if the spectral differences $\Delta s_{jm}^{[q]}=s_{jm}^{[q]}-\bar{s}_{m}^{[q]}$ satisfy 
\[
	\bigl|\Delta s_{jm}^{[q]}\bigr| \le \delta s_{jm}^{[q]} \quad \text{for all indices } q,j,m,
\]
then for every $q$, $d$, $j$, the following inequality holds
\[
	\bigl|w_{dj}^{[q]}(\hat{\mathbf{f}})-\phi_{d}^{[q]}\bigr| \leq -\epsilon\phi_{d}^{[q]}.
\]
\end{lemma}
	\begin{proof}
		For brevity, define
		\begin{align*}
				A_m^{[q]} &= \exp\Biggl(-\sum_{d=1}^{D} b_{dm} C_d^{[q]}\Biggr), \\
				E_{d}^{[q]} (\mathbf{s}) &= \sum_{m=1}^{M} b_{dm}s_m A_m^{[q]}>0, \\
				 F^{[q]} (\mathbf{s}) &= \sum_{m=1}^{M} s_m A_m^{[q]}>0,
		\end{align*}
		where $\mathbf{s} = [s_1,\dots,s_M]\trans$.
		Recall the definitions that $w_{dj}^{[q]}(\hat{\mathbf{f}}) = -E_d^{[q]}(\mathbf{s}_j^{[q]})/F^{[q]}(\mathbf{s}_j^{[q]})$ and $\phi_d^{[q]} = -E_d^{[q]}(\bar{\mathbf{s}}^{[q]})/F^{[q]}(\bar{\mathbf{s}}^{[q]})$.
		We first establish bounds that hold uniformly for all indices. For the denominator difference, using the given condition,
		\begin{align*}
				\bigl|F^{[q]}(\mathbf{s}_j^{[q]}) - F^{[q]}(\bar{\mathbf{s}}^{[q]}) \bigr| &= 
				\left| \sum_{m=1}^{M} \Delta s_{jm}^{[q]} A_m^{[q]} \right| \\
				& \leq \delta \sum_{m=1}^{M} s_{jm}^{[q]} A_m^{[q]} = \delta F^{[q]}(\mathbf{s}_j^{[q]}),
		\end{align*}
		where the inequalities hold for all $q$, $j$. Similarly, for the numerator difference,
		\[
			\bigl|E_d^{[q]}(\mathbf{s}_j^{[q]}) - E_d^{[q]}(\bar{\mathbf{s}}^{[q]}) \bigr|\leq \delta E_d^{[q]}(\mathbf{s}_j^{[q]}), \quad \forall  q, d, j.
		\]
		Using the triangle inequality and the fact that $\delta<1$, we obtain for all $q$, $d$, $j$,
		\[
			E_d^{[q]}(\bar{\mathbf{s}}^{[q]}) \geq (1-\delta) E_d^{[q]}(\mathbf{s}_j^{[q]}).
		\]
		Next we compute the relative difference, for any $q$, $d$, $j$,
		\[
				\frac{| w_{dj}^{[q]}(\hat{\mathbf{f}})  -  \phi_{d}^{[q]}|}{ |\phi_{d}^{[q]}|}  
				\leq \frac{2\delta E_d^{[q]}(\mathbf{s}_j^{[q]})}{(1-\delta)E_d^{[q]}(\mathbf{s}_j^{[q]})}=\frac{2\delta}{1-\delta}.
		\]
		Thus, taking $\epsilon=2\delta/(1-\delta)$, we have for all $q$, $d$, $j$
		\[
			\bigl|w_{dj}^{[q]}(\hat{\mathbf{f}})  -  \phi_{d}^{[q]}   \bigr | \leq  -\epsilon \phi_{d}^{[q]}.
		\]
		This completes the proof.
	\end{proof}
	
	This lemma establishes that if the spectral differences between the ray-dependent energy spectra $\{\mathbf{s}_{j}^{[q]}\}_{j=1}^{J^{[q]}}$ and their aggregated version $\bar{\mathbf{s}}^{[q]}$ is sufficiently small, then the relative difference between the weighting coefficients involved in the approximate Jacobian matrix $\mathbf{J}(\hat{\mathbf{f}})$ and those in the exact Jacobian $D \mathbf{K}(\hat{\mathbf{f}})$ is bounded. Building on this lemma, we can further derive the following result.
	
	\begin{lemma}\label{lemma:jacobian_diff}
		Assume that $\mathbf{J}(\hat{\mathbf{f}})$ has full column rank. Let $\delta\in(0,1)$ and suppose the spectral differences $\Delta s_{jm}^{[q]}$ satisfy the condition in \cref{lemma:weight_diff}. Then there exists a constant $\gamma>0$ such that for any $\mathbf{h}$,
		\[
		\big\|D \mathbf{K}(\hat{\mathbf{f}})\mathbf{h}-\mathbf{J}(\hat{\mathbf{f}})\mathbf{h}\big\|_2
		\leq \gamma\big\|\mathbf{J}(\hat{\mathbf{f}})\mathbf{h}\big\|_2.
		\]
	\end{lemma}
	\begin{proof}
		By \cref{lemma:weight_diff}, there exists a constant $\epsilon>0$ such that for all $q$, $d$, $j$,
		\[
		\bigl|w_{dj}^{[q]}(\hat{\mathbf{f}})-\phi_{d}^{[q]}\bigr| \leq -\epsilon \phi_{d}^{[q]}.
		\]
		Using this bound, we obtain
		\begin{align*}
				\big\|D \mathbf{K}(\hat{\mathbf{f}})-\mathbf{J}(\hat{\mathbf{f}})\big\|_F^2
				& =\sum_{q=1}^{Q}\sum_{j=1}^{J^{[q]}}\sum_{d=1}^{D} \big(w_{dj}^{[q]}(\hat{\mathbf{f}})-\phi_{d}^{[q]}\big)^2 \big\|\mathbf{p}_{j}^{[q]}\big\|_2^2 \\
				& \leq \epsilon^2 \big\| \mathbf{J}(\hat{\mathbf{f}})\big\|_F^2.
		\end{align*}
		Applying standard matrix norm inequalities,
		\begin{align*}
				\big\|D \mathbf{K}(\hat{\mathbf{f}})\mathbf{h}-\mathbf{J}(\hat{\mathbf{f}})\mathbf{h}\big\|_2 
				& \leq \big\|D \mathbf{K}(\hat{\mathbf{f}})-\mathbf{J}(\hat{\mathbf{f}})\big\|_F \|\mathbf{h}\|_2 \\
				& \leq \epsilon\kappa_F\bigl(\mathbf{J}(\hat{\mathbf{f}})\bigr) \big\|\mathbf{J}(\hat{\mathbf{f}})\mathbf{h}\big\|_2.
		\end{align*}
		Letting $\gamma = \epsilon\kappa_F\bigl(\mathbf{J}(\hat{\mathbf{f}})\bigr)$, we complete the proof.
	\end{proof}
	
	The preceding lemma implies that, provided the spectral difference is sufficiently small, there exists an upper bound on the relative difference between the approximate Jacobian matrix $\mathbf{J}(\hat{\mathbf{f}})$ and the exact Jacobian matrix $D \mathbf{K}(\hat{\mathbf{f}})$. 	
	
	We next establish an upper bound on the relative difference between the exact Jacobian matrices evaluated at distinct points.
	
	\begin{lemma}\label{lemma:jacobian_exact_diff}
		Assume that $D\mathbf{K}(\hat{\mathbf{f}})$ has full column rank. Then there exists a constant $\eta>0$ such that for any $\mathbf{f}$ and $\mathbf{h}$,
		\[
		\big\|D\mathbf{K}(\hat{\mathbf{f}})\mathbf{h}-D\mathbf{K}(\mathbf{f})\mathbf{h}\big\|_2
		\leq \eta\big\|D\mathbf{K}(\hat{\mathbf{f}})\mathbf{h}\big\|_2.
		\]
	\end{lemma}
	\begin{proof}
		For any $\mathbf{f}$, the values $w_{dj}^{[q]}(\mathbf{f})$ satisfy
		\[
			-\max_m \{b_{dm}\} \leq w_{dj}^{[q]}(\mathbf{f}) \leq -\min_m \{b_{dm}\}.
		\]
		Then it follows that
		\[
			\bigl|w_{dj}^{[q]}(\hat{\mathbf{f}})-w_{dj}^{[q]}(\mathbf{f})\bigr| \leq \Delta_{d},
		\]
		where $\Delta_{d} = \max_m \{b_{dm}\} - \min_m \{b_{dm}\}$. Using this bound, we obtain
		\[
				\big\|D\mathbf{K}(\hat{\mathbf{f}})-D\mathbf{K}(\mathbf{f})\big\|_F^2
				 \leq \left(\sum_{d=1}^{D} \Delta_{d}^2\right) \|\mathbf{P}\|_F^2 .
		\]
		Then for any $\mathbf{f}$ and $\mathbf{h}$,
		\begin{align*}
				\big\|D \mathbf{K}(\hat{\mathbf{f}})\mathbf{h}-D\mathbf{K}(\mathbf{f})\mathbf{h}\big\|_2 
				 &\leq \big\|D \mathbf{K}(\hat{\mathbf{f}})-D\mathbf{K}(\mathbf{f})\big\|_F \|\mathbf{h}\|_2 \\
				 &\leq \sqrt{\sum_{d=1}^{D} \Delta_{d}^2} \|\mathbf{P}\|_F \sigma_{\min}^{-1}\bigl(D\mathbf{K}(\hat{\mathbf{f}})\bigr) \big\|D\mathbf{K}(\hat{\mathbf{f}})\mathbf{h}\big\|_2.
		\end{align*}
		Letting $\eta = \sqrt{\sum_{d=1}^{D} \Delta_{d}^2} \|\mathbf{P}\|_F \sigma_{\min}^{-1}\bigl(D\mathbf{K}(\hat{\mathbf{f}})\bigr)$, we complete the proof.
	\end{proof}
	
	In practical computations, the iterative scheme \eqref{sni3} may introduce both approximation and numerical errors. At each iteration, the residual can be defined as
	\begin{equation}\label{res1}
		\bm{\delta}^k = \mathbf{J}(\hat{\mathbf{f}}) (\mathbf{f}^{k+1}-\mathbf{f}^{k})+\big(\mathbf{K}(\mathbf{f}^{k})-\mathbf{g}\big).
	\end{equation}
	Under appropriate conditions, we have the following global convergence result for the iterative scheme \eqref{sni3}.
	
	\begin{theorem}
	Assume that the conditions in \cref{lemma:jacobian_exact_diff} hold and the $\eta < 1$. Then the solution to \eqref{nf2} is unique.
	
	Moreover, suppose that the conditions in \cref{lemma:jacobian_diff} hold, and for every $k$, the residual in \eqref{res1} satisfies
		\begin{equation}\label{eq:delta_estimate}
			\|\bm{\delta}^k\|_2 \leq \zeta \|\mathbf{K}(\mathbf{f}^{k})-\mathbf{g}\|_2,
		\end{equation}
		with the constants satisfying
		\[
			(1+\zeta)(1+\gamma)(1+\eta)<2.
		\]
		Then the sequence $\{\mathbf{f}^k\}$ generated by the iterative scheme \eqref{sni3} converges to the unique solution of \eqref{nf2}.
	\end{theorem}
	\begin{proof}
		By the mean value theorem and \cref{lemma:jacobian_exact_diff}, for any $\mathbf{f}$, it holds that 
		\begin{align}\label{eq:estimate}
			\nonumber	\big\|\mathbf{K}(&\mathbf{f}) - \mathbf{K}(\mathbf{f}^\ast) - D \mathbf{K}(\hat{\mathbf{f}})(\mathbf{f} - \mathbf{f}^\ast)   \big\|_2   \\
			\nonumber &\leq  \int_{0}^{1} \big\| D \mathbf{K}\bigl(t\mathbf{f}+(1-t)\mathbf{f}^\ast\bigr)(\mathbf{f} - \mathbf{f}^\ast) - D \mathbf{K}(\hat{\mathbf{f}})(\mathbf{f} - \mathbf{f}^\ast) \big\|_2 \diff t \\
				&\leq  \eta\big\| D \mathbf{K}(\hat{\mathbf{f}})(\mathbf{f} - \mathbf{f}^\ast) \big\|_2.
		\end{align}
We now prove uniqueness by contradiction. Suppose that there exists another solution $\mathbf{f}^{\prime}$.
		Substituting it into the above inequality yields
		\[
			\big\|   D \mathbf{K}(\hat{\mathbf{f}})(\mathbf{f}^{\prime} - \mathbf{f}^\ast) \big\|_2 
			\leq \eta\big\| D \mathbf{K}(\hat{\mathbf{f}})(\mathbf{f}^{\prime} - \mathbf{f}^\ast) \big\|_2.
		\]
		Since $\eta<1$, we conclude that
		\[
			\big\|   D \mathbf{K}(\hat{\mathbf{f}})(\mathbf{f}^{\prime} - \mathbf{f}^\ast) \big\|_2=0.
		\]
		As $D \mathbf{K}(\hat{\mathbf{f}})$ has full column rank, it follows that $\mathbf{f}^{\prime} = \mathbf{f}^\ast$,
		which establishes the uniqueness of the solution.
		
		We next derive the contraction property. For \eqref{eq:estimate}, using the triangle inequality, we obtain
		\begin{equation}\label{eq:estimate2}
			\big\| \mathbf{K}(\mathbf{f}) - \mathbf{K}(\mathbf{f}^\ast) \big\|_2  \leq (1+\eta) \big\| D \mathbf{K}(\hat{\mathbf{f}})(\mathbf{f} - \mathbf{f}^\ast) \big\|_2.
		\end{equation}
		By adding and subtracting the term $\mathbf{J}(\hat{\mathbf{f}})(\mathbf{f}^{k+1}-\mathbf{f}^k)$, applying the triangle inequality repeatedly, and using \eqref{res1}-\eqref{eq:estimate2}, we have
		\begin{equation}\label{theorem:estimated}
			 \big\| D \mathbf{K}(\hat{\mathbf{f}})(\mathbf{f}^{k+1} - \mathbf{f}^\ast) \big\|_2 
				\leq  \big((1+ \zeta)(1 + \gamma)(1 + \eta) - 1\big) \big\| D \mathbf{K}(\hat{\mathbf{f}})(\mathbf{f}^k  - \mathbf{f}^\ast) \big\|_2.
		\end{equation}
		Since $(1+\zeta)(1+\gamma)(1+\eta) - 1 < 1$,
		the sequence $\big\{ \big\| D\mathbf{K}(\hat{\mathbf{f}})(\mathbf{f}^{k}-\mathbf{f}^\ast) \big\|_2 \big\}$ converges to zero.
		Combined with the full column rank condition of $D \mathbf{K}(\hat{\mathbf{f}})$,
		the iterative sequence $\{\mathbf{f}^k\}$ converges to $\mathbf{f}^\ast$.
		This completes the proof.
	\end{proof}
	
	Furthermore, for the iterative scheme \eqref{sni3}, based on \eqref{theorem:estimated}, we have 
	\begin{multline*}
			\sigma_{\min}\bigl(D \mathbf{K}(\hat{\mathbf{f}})\bigr) \|\mathbf{f}^{k+1} - \mathbf{f}^\ast\|_2 \leq \big\| D \mathbf{K}(\hat{\mathbf{f}})(\mathbf{f}^{k+1} - \mathbf{f}^\ast) \big\|_2 \\
			\leq \sigma_{\max}\bigl(D \mathbf{K}(\hat{\mathbf{f}})\bigr) \big((1+\zeta)(1+\gamma)(1+\eta)-1\big) \| \mathbf{f}^k - \mathbf{f}^\ast \|_2. 
	\end{multline*}
Then
	\[
		\| \mathbf{f}^{k+1} - \mathbf{f}^\ast \|_2 \leq \kappa\bigl(D \mathbf{K}(\hat{\mathbf{f}})\bigr) \big((1+\zeta)(1+\gamma)(1+\eta)-1\big) \| \mathbf{f}^k - \mathbf{f}^\ast \|_2,
	\]
	where $\kappa(\cdot)=\sigma_{\max}(\cdot)/\sigma_{\min}(\cdot)$ is the condition number. 
	
	Note that if $\kappa\bigl(D \mathbf{K}(\hat{\mathbf{f}})\bigr) \big((1+\zeta)(1+\gamma)(1+\eta)-1\big)\in(0,1)$, the proposed algorithm converges Q-linearly. This also shows that a smaller residual $\bm{\delta}^k$  yields a smaller constant $\zeta$  and thus a faster convergence rate.

	\section{Numerical experiments}\label{section5}
	We perform numerical experiments to validate the performance of the proposed algorithm for solving geometric-inconsistent MSCT reconstruction problem with ray-dependent energy spectra. For comparison, we include several existing iterative methods, namely,  the NKM with cyclic strategy \cite{Gao_NKM}, and the NCPD algorithm \cite{CHEN2021101821}. Here, for NKM, one iteration corresponds to a complete cycle of updates.
	
	\subsection{Experiment settings}
	The tests are conducted on a workstation running Python with Operator Discretization Library (ODL) \cite{adler2018odl}, equipped with an Intel i9-13900K @3.0GHz CPU and Nvidia RTX A4000 GPU. In experiments, a dual‑energy spectrum configuration is employed: the low‑energy spectrum is set at 80 kV, and the high‑energy spectrum at 140 kV with an additional 1 mm thick copper filter. Both spectra are generated using the open‑source spectrum simulator SpectrumGUI \cite{spectrumgui}, with a discrete energy sampling interval of 1 kV. The ray-dependent energy spectra are obtained by introducing perturbations to the low- and high-energy spectra followed by non‑negative normalization. Part of the ray-dependent energy spectra used in our tests are shown in \cref{spectra}. The MACs of the basis materials, namely, water and bone, are obtained from \cite{Hubbell1995TablesOX}. Note that using two basis materials is a standard and physically justified approach in clinical DECT. Given that both fan-beam and cone-beam geometries can be equivalently transformed into a parallel-beam configuration, we employ geometric-inconsistent parallel-beam geometries for the low- and high-energy spectra.
	
With these configurations, we generate simulated noiseless data by substituting the ray-dependent energy spectra, MACs of water and bone, discrete X-ray transform, and truth basis images into \eqref{dm2}. Reconstruction accuracy is assessed by using the following relative error metrics
	\[
		\text{RE}_{\mathbf{f}}^{k}:=\frac{\|\mathbf{f}^{k}-\mathbf{f}^{\ast}\|_2}{\|\mathbf{f}^{\ast}\|_2},\quad \text{RE}_{\mathbf{g}}^{k}:=\frac{\|\mathbf{K}(\mathbf{f}^{k})-\mathbf{g}\|_2}{\|\mathbf{g}\|_2},
	\]
where $\mathbf{f}^{k}$ denotes the $k$-th iteration point computed by the respective algorithm, $\mathbf{f}^{\ast}$ the ground truth, and $\mathbf{g}$ the simulated noiseless or noisy data. Furthermore, when the measured data contains noise, we use the following error metrics between adjacent iterations to evaluate the numerical convergence of the iterative algorithm
	\[
		\Delta_{\mathbf{f}}^{k}:=\frac{\|\mathbf{f}^{k+1}-\mathbf{f}^{k}\|_2}{\|\mathbf{f}^{k}\|_2},\quad \Delta_{\mathbf{g}}^k:=\frac{\|\mathbf{K}(\mathbf{f}^{k+1})-\mathbf{K}(\mathbf{f}^{k})\|_2}{\|\mathbf{g}\|_2}.
	\]
	
In \cref{alg:msct_recon}, we set $C_d^{[q]}=0$, $q=1,\dots,Q$, $d=1,\dots,D$ by selecting $\hat{\mathbf{f}}=\mathbf{0}$. We implement the $(\mathbf{P}^{[q]})^{+}$, $q=1,\dots,Q$ using the FBP method with the 1-bandwidth Ram--Lak filter. Without loss of generality, the initial points for all tested algorithms are set to the zero vector.
	\begin{figure}[htbp]
		\centering
		\includegraphics[width=0.9\textwidth]{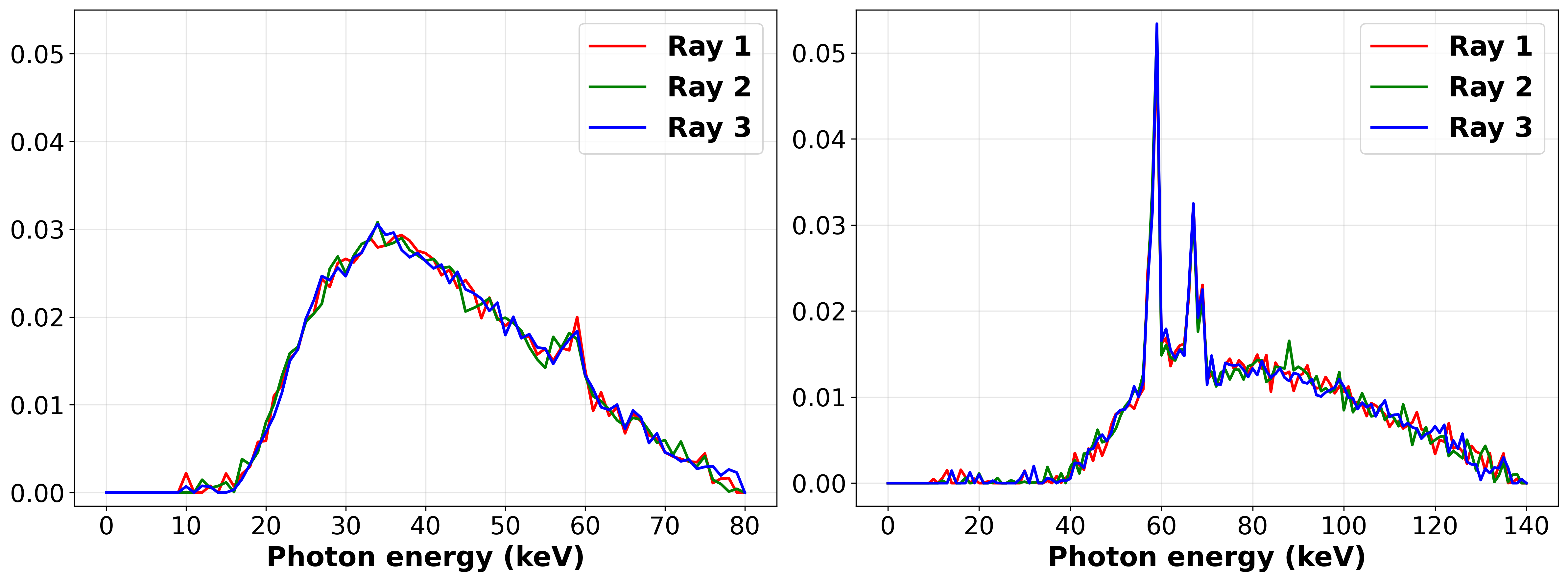}
		\caption{Part of the ray-dependent energy spectra used in the tests. The left figure shows the 80 kV spectra, and the right figure shows the 140 kV spectra.}
		\label{spectra}
	\end{figure}
	
	\subsection{Simulated noiseless data for the forbild phantom}
	In this test, a forbild phantom is composed of the truth water and bone basis images which are both discretized as $128\times128$ pixels on the square area $\left[-5,5\right] \times \left[-5,5\right]\text{cm}^2$ and take gray values over $\left[0,1\right]$, as shown in row 1 of \cref{result1}.
	
The scanning parameters are as follows: for the high-energy spectrum, 384 scanning angles are uniformly distributed over $\left[0,\pi\right)$, with 384 detector elements evenly spaced over the interval $\left[-7.05, 7.05\right]\text{cm}$. For the low-energy spectrum, 384 scanning angles are uniformly distributed over $\left[\pi/768,\pi/768 + \pi\right)$, while the detector configuration remains the same as that used for the high-energy case.
	
All algorithms are executed over 100 iterations. The curves of the performance metrics $\text{RE}_{\mathbf{f}}^{k}$ and $\text{RE}_{\mathbf{g}}^{k}$ as function of the iteration number are depicted in \cref{RE_iter}. It can be clearly observed that the proposed algorithm attains stable convergence with high accuracy within only 60 iterations, whereas the relative errors associated with the NCPD and NKM algorithms decrease slowly. To further demonstrate the computational efficiency of the proposed algorithm, we also present the relative error metrics $\text{RE}_{\mathbf{f}}^{k}$ and $\text{RE}_{\mathbf{g}}^{k}$ plotted over time. As illustrated in \cref{RE_time}, for $128 \times 128$ image reconstruction tasks, the proposed algorithm achieves machine precision in just over ten seconds, whereas the other two iterative methods require substantially more time to converge.

	\begin{figure}[htbp]
		\centering
		\includegraphics[width=0.9\textwidth]{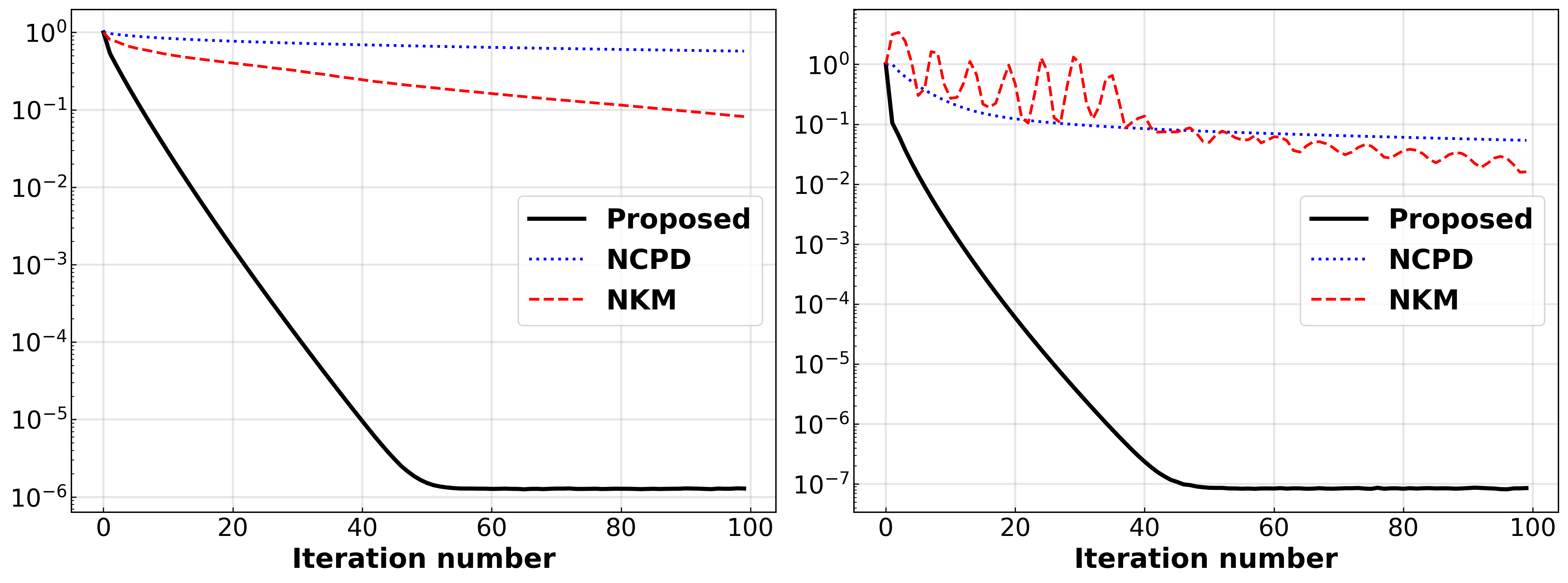}
		\caption{Metrics $\text{RE}_{\mathbf{f}}^{k}$ (left) and $\text{RE}_{\mathbf{g}}^{k}$ (right) are plotted in semi-log scale over iteration numbers for reconstructing basis images of the forbild phantom by the proposed algorithm, NCPD and NKM.}
		\label{RE_iter}
	\end{figure}
	
	\begin{figure}[htbp]
		\centering
		\includegraphics[width=0.9\textwidth]{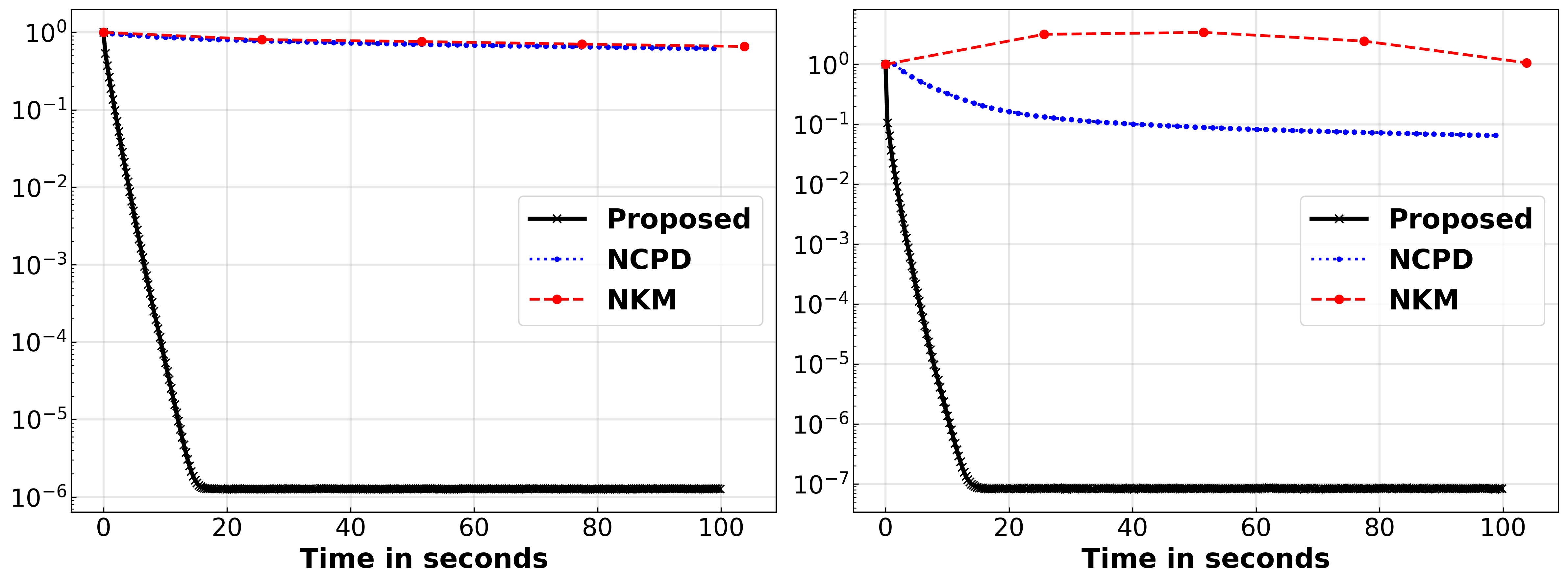}
		\caption{Metrics $\text{RE}_{\mathbf{f}}^{k}$ (left) and $\text{RE}_{\mathbf{g}}^{k}$ (right) are plotted in semi-log scale over time (in seconds) for reconstructing basis images of the forbild phantom by the proposed algorithm, NCPD and NKM.}
		\label{RE_time}
	\end{figure}
	
These results demonstrate that the proposed algorithm enables efficient and high-precision inversion for noiseless data. The reconstructed basis images yielded by all methods, together with the VMIs and their corresponding ground truth references, are displayed in \cref{result1}. Visual inspection intuitively confirms that the reconstruction outcomes produced by the proposed algorithm agree most closely with the ground truth.
	
	\begin{figure}[htbp]
		\centering
		\includegraphics[width=0.9\textwidth]{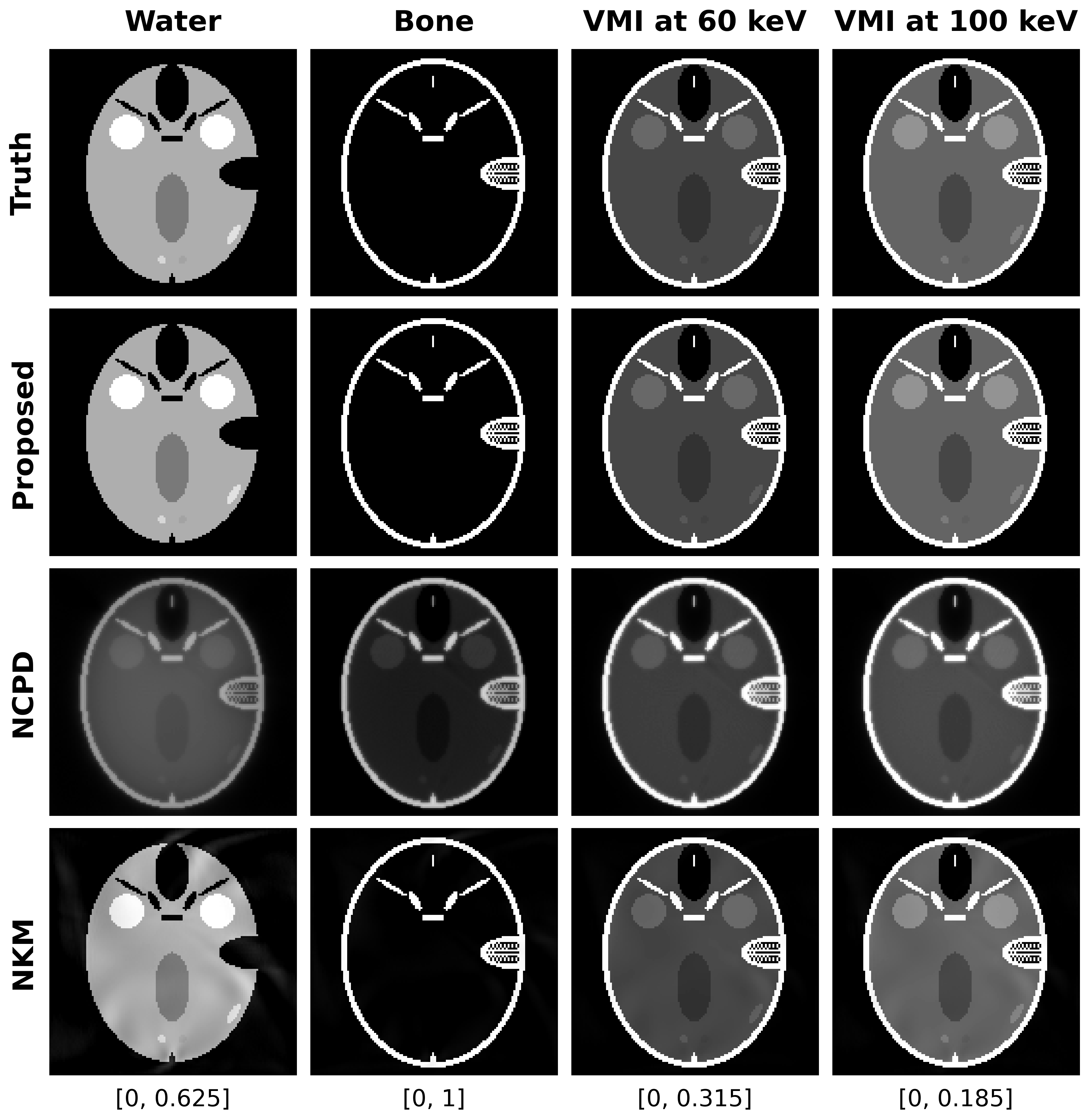}
		\caption{Reconstruction results in the noiseless case. From left to right: basis images of water and bone, VMIs at energies 60 keV and 100 keV of the forbild phantom. From top to bottom: the true images, the results after 100 iterations using the proposed algorithm, NCPD and NKM.}
		\label{result1}
	\end{figure}
	
	\subsection{Simulated noisy data for the realistic torso phantom}
	To further validate the robustness of the proposed algorithm on noisy data, we conduct a numerical simulation based on a realistic torso image of size $256 \times 256$, defined over the square domain $\left[-5,5\right] \times \left[-5,5\right]\text{cm}^2$ with gray values normalized to the range $\left[0,1\right]$, as shown in the first row of \cref{result2}.
	
We adopt a similar parallel-beam geometric configuration. For the high-energy spectrum, 768 detectors are uniformly distributed over $\left[-7.05, 7.05\right]\text{cm}$, with 768 scanning angles evenly spaced over $\left[0,\pi\right)$. For the low-energy spectrum, 768 scanning angles are uniformly distributed over $\left[\pi/1536,\pi/1536+\pi\right)$, while the detector setup remains identical to the high-energy case. Moreover, noisy data with 27.2 dB are generated by adding Gaussian noise to the noiseless data.
	
All algorithms are executed over 50 iterations. We plot the curves of $\text{RE}_{\mathbf{f}}^{k}$ and $\text{RE}_{\mathbf{g}}^{k}$ for all algorithms in \cref{RE_iter_256}. The proposed algorithm achieves numerical convergence within 10 iterations, whereas the NCPD and NKM algorithms do not converge within 50 iterations. This verifies its robustness and efficiency in addressing noisy data. The curves of $\text{RE}_{\mathbf{f}}^{k}$ and $\text{RE}_{\mathbf{g}}^{k}$ over time are displayed in \cref{RE_time_256}. It is observed that the proposed algorithm can attain high accuracy for noisy $256 \times 256$ image reconstruction within only 20 seconds. 
	
To further illustrate the convergence behavior under noisy conditions, we plot the curves of metrics $\Delta_{\mathbf{f}}^{k}$ and $\Delta_{\mathbf{g}}^{k}$ in \cref{DE_iter_256}. As shown, the proposed algorithm attains machine precision in metrics between adjacent iterations within 40 iterations, verifying its efficient numerical convergence in noisy conditions.
	
	\begin{figure}[htbp]
		\centering
		\includegraphics[width=0.9\textwidth]{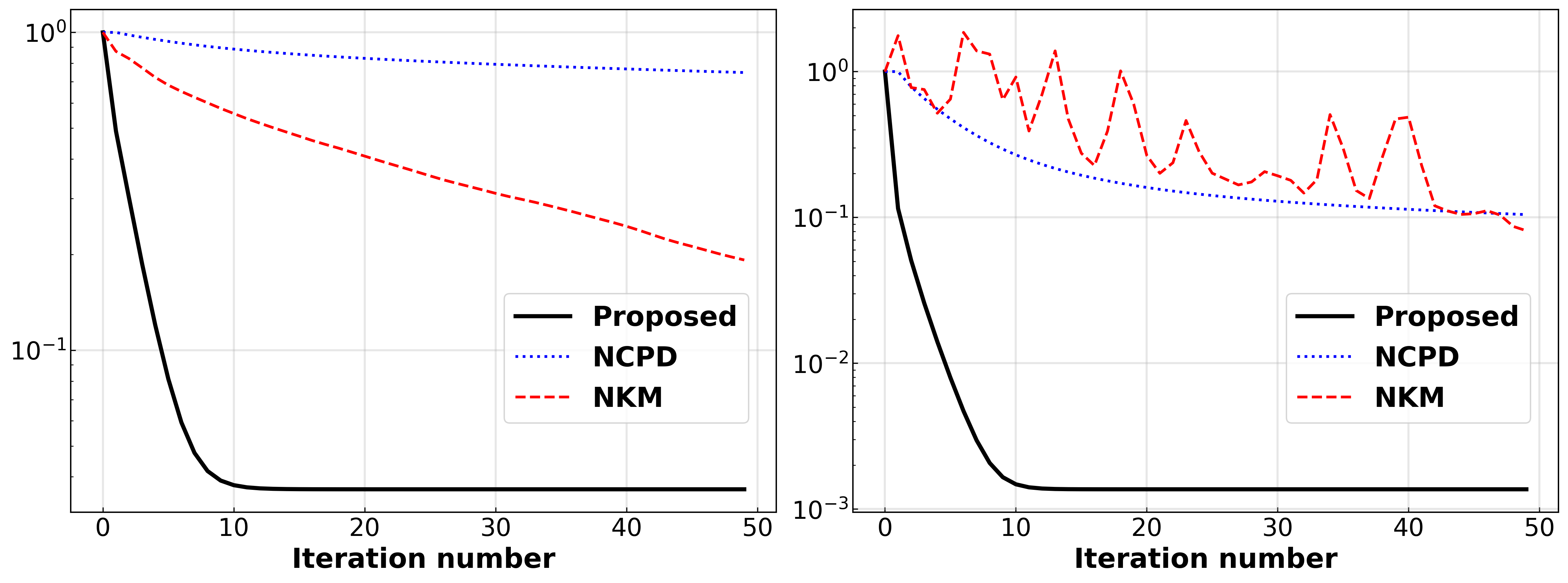}
		\caption{Metrics $\text{RE}_{\mathbf{f}}^{k}$ (left) and $\text{RE}_{\mathbf{g}}^{k}$ (right) are plotted in semi-log scale over iteration numbers for reconstructing basis images of the realistic torso image by the proposed algorithm, NCPD and NKM.}
		\label{RE_iter_256}
	\end{figure}
	
	\begin{figure}[htbp]
		\centering
		\includegraphics[width=0.9\textwidth]{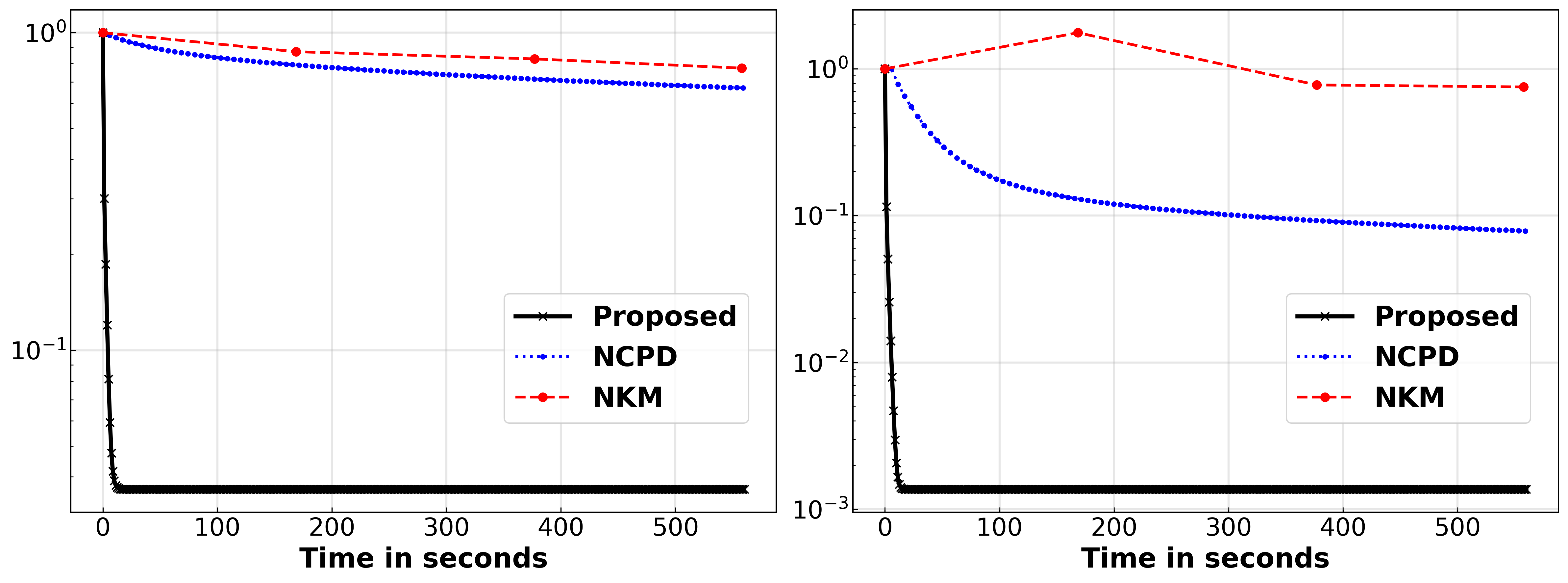}
		\caption{Metrics $\text{RE}_{\mathbf{f}}^{k}$ (left) and $\text{RE}_{\mathbf{g}}^{k}$ (right) are plotted in semi-log scale over time (in seconds) for reconstructing basis images of the realistic torso image by the proposed algorithm, NCPD and NKM.}
		\label{RE_time_256}
	\end{figure}
	
	\begin{figure}[htbp]
		\centering
		\includegraphics[width=0.9\textwidth]{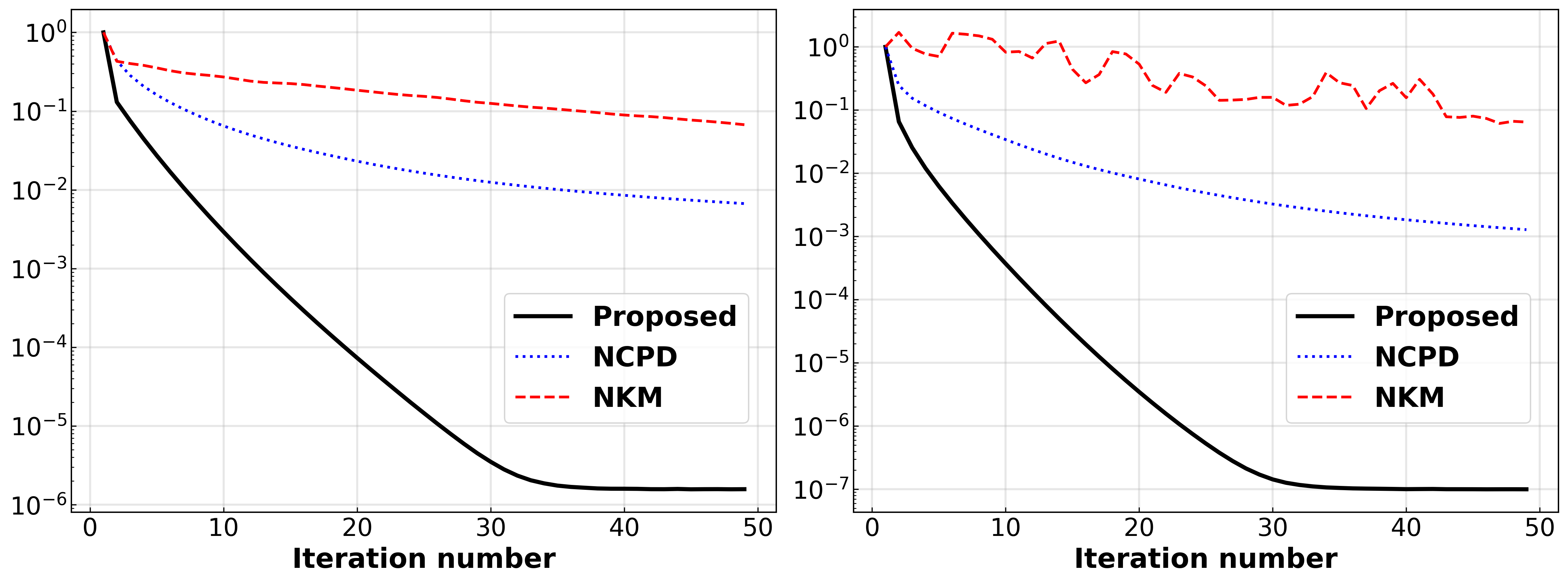}
		\caption{Metrics $\Delta_{\mathbf{f}}^{k}$ (left) and $\Delta_{\mathbf{g}}^{k}$ (right) are plotted in semi-log scale over iteration numbers for reconstructing basis images of the realistic torso image by the proposed algorithm, NCPD and NKM.}
		\label{DE_iter_256}
	\end{figure}
	
The reconstructed basis images, VMIs and their corresponding ground truth references obtained by all methods are shown in \cref{result2}, while the absolute difference maps are displayed in \cref{Result_Difference_256}. From these difference maps, it can be observed that the reconstruction results of the proposed algorithm agree most closely with the ground truth.
	
	\begin{figure}[htbp]
		\centering
		\includegraphics[width=0.9\textwidth]{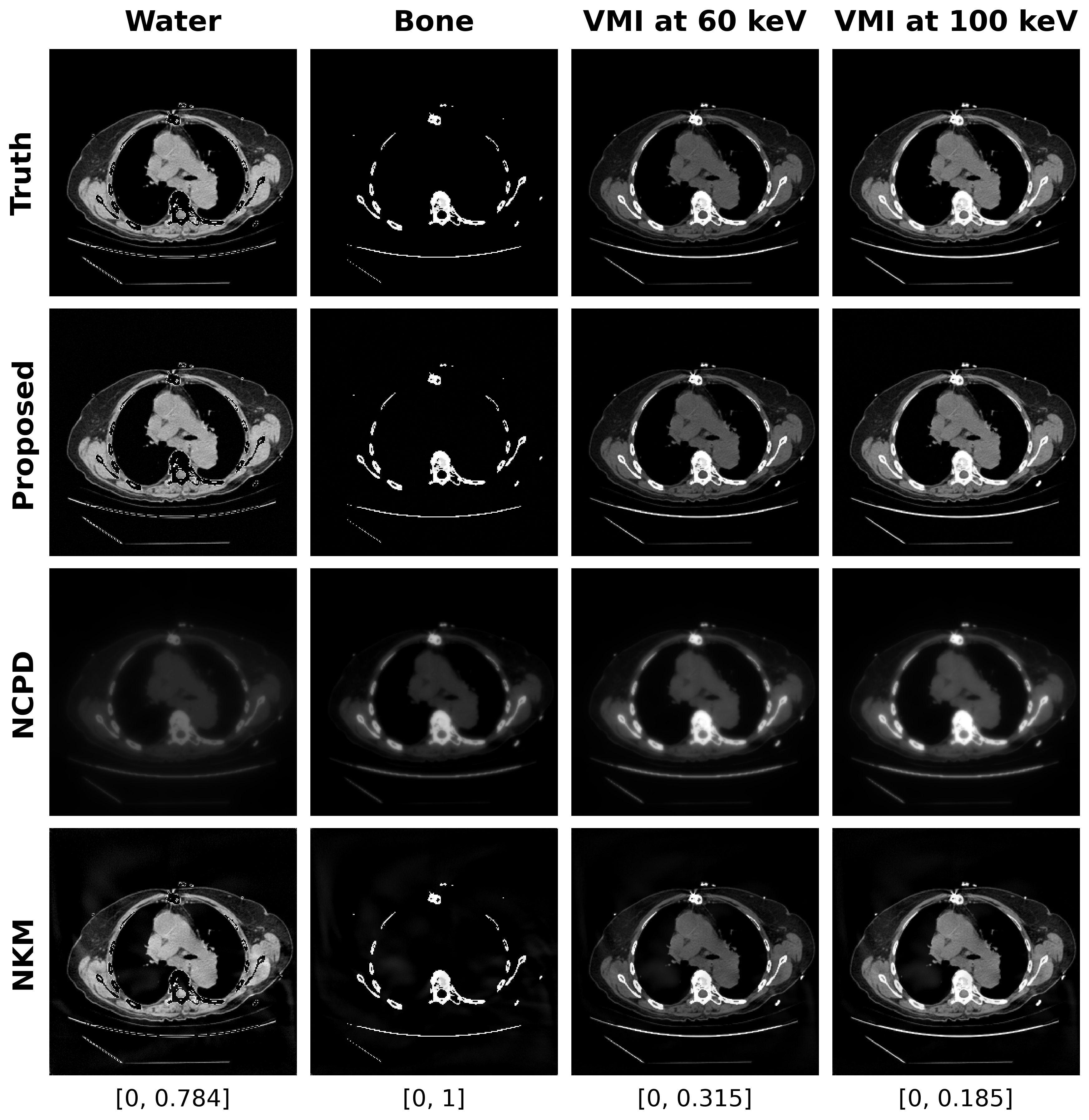}
		\caption{Reconstruction results in the noisy case. From left to right: basis images of water and bone, VMIs at energies 60 keV and 100 keV of the realistic torso image. From top to bottom: the true images, the results after 50 iterations using the proposed algorithm, NCPD and NKM.}
		\label{result2}
	\end{figure}

	\begin{figure}[htbp]
		\centering
		\includegraphics[width=0.9\textwidth]{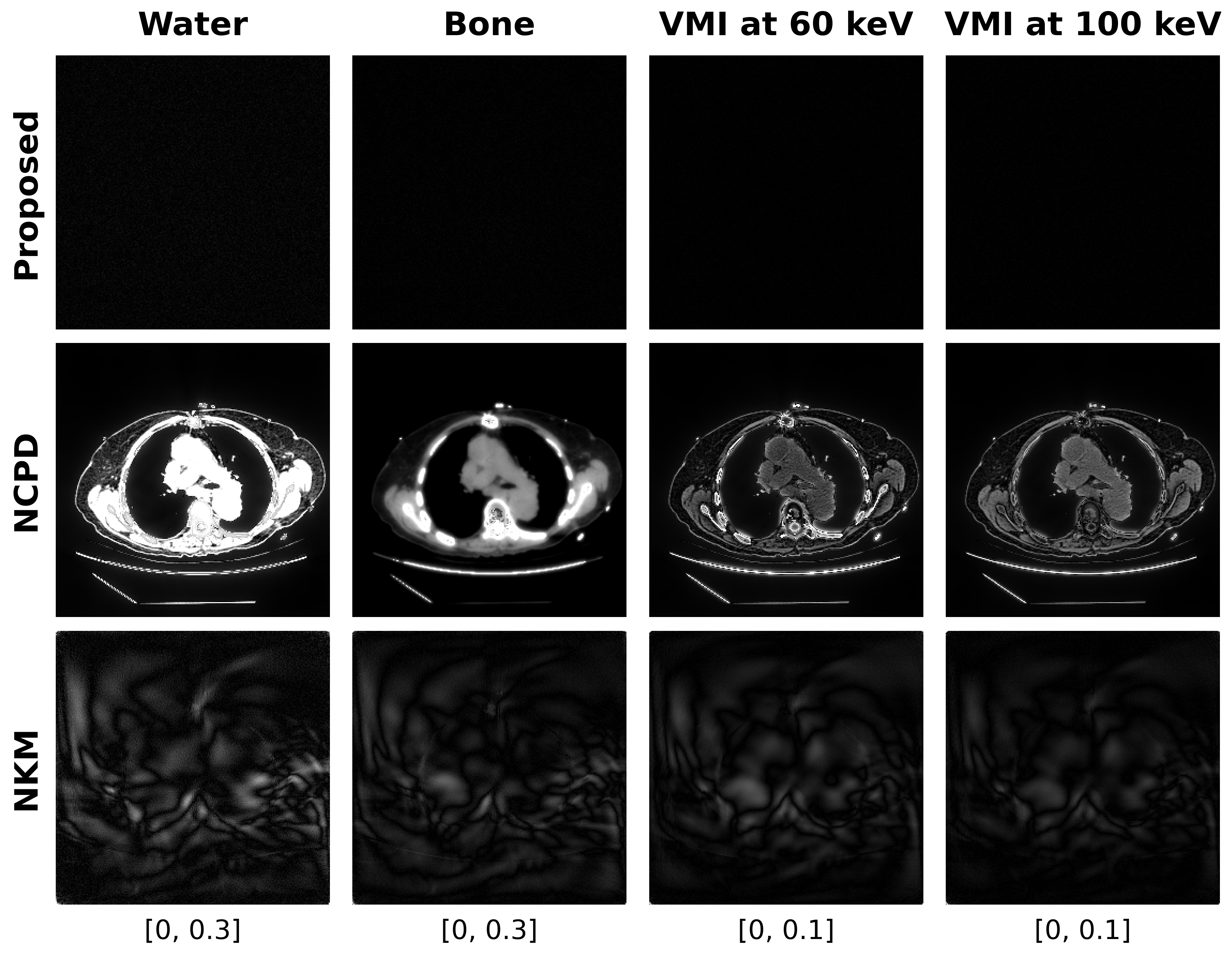}
		\caption{Absolute difference maps between the reconstructed results and the ground truths in the noisy case. Each column from left to right corresponds to the water basis image, bone basis image, VMI at 60 keV, and VMI at 100 keV, respectively. Each row from top to bottom shows the absolute difference maps for the proposed algorithm, NCPD, and NKM, rspectively.}
		\label{Result_Difference_256}
	\end{figure}

	\subsection{Alternative selection methods for aggregated energy spectrum}
	In fact, there are many viable ways to construct an aggregated energy spectrum, as long as the ray-dependent spectrum is decoupled from the rays while the overall spectral shape is roughly preserved. For instance, the median-aggregated energy spectrum can be adopted, defined as
	\[
		\bar{s}_{m}^{[q]}= \text{median}\{s_{jm}^{[q]}\}_{j=1}^{J^{[q]}},
	\]
	and the $L_2$-mean-aggregated energy spectrum can be similarly defined as
	\[
		\bar{s}_{m}^{[q]} = \sqrt{\frac{1}{J^{[q]}}\sum_{j=1}^{J^{[q]}} (s_{jm}^{[q]})^2 }.
	\]
	Note that if the energy spectrum is ray-independent, all of the three aggregated counterparts reduce to the same spectrum itself. 
	
	To validate the feasibility of these different selection approaches, we carried out a numerical experiment. The detailed configurations are outlined below. We employed realistic water and bone basis images randomly selected from the LiTS dataset \cite{LiST}, as displayed in the first row of \cref{result3}. Each basis image is discretized into $512\times512$ pixels over a square field of view spanning $[-15, 15] \times [-15, 15] \text{ cm}^2$. 

In this experiment, under the 80 kV spectrum, 1280 parallel projections were uniformly sampled over a detector range of $[-21.15, 21.15] \text{ cm}$ for each of the $1280$ views, with the views evenly distributed across $[0,\pi)$. For the filtered 140 kV spectrum, the same numbers of projections and views were adopted, but the view angles  are shifted to the interval $[\pi/2560, \pi/2560 + \pi)$.
	
We denote the algorithm employing the mean-aggregated energy spectrum as Mean, the algorithm with the median-aggregated energy spectrum as Median, and the algorithm with the $L_2$-mean-aggregated energy spectrum as $L_2$-Mean. The initial point is initialized as $\mathbf{f}^0=\mathbf{0}$.
	
As shown in \cref{RE_iter_512}, we plot the curves of metrics $\text{RE}_{\mathbf{f}}^{k}$ and $\text{RE}_{\mathbf{g}}^{k}$ over 100 iterations. It is observed that the convergence behaviors of the proposed algorithm using these three different aggregated energy spectra are almost identical. The corresponding basis images and VMIs are displayed in \cref{result3}, and they are visually comparable and all closely approximate the true images.
	
Collectively, these findings demonstrate that the convergence, accuracy, and efficiency of the proposed algorithm remain robust against the specific choice of the aggregated energy spectra.
	
	\begin{figure}[htbp]
		\centering
		\includegraphics[width=0.9\textwidth]{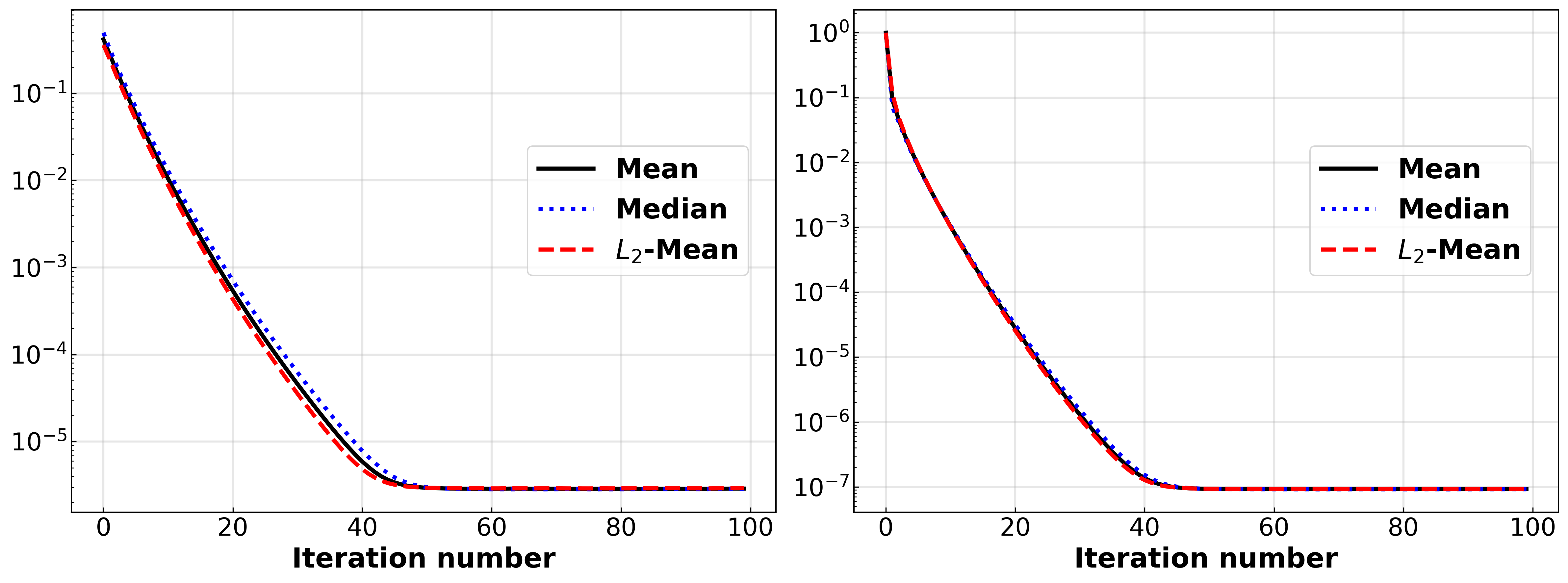}
		\caption{Metrics $\text{RE}_{\mathbf{f}}^{k}$ (left) and $\text{RE}_{\mathbf{g}}^{k}$ (right) are plotted in semi-log scale over iteration numbers for reconstructing basis images of from \cref{result3} by the proposed algorithm with different aggregated energy spectra.}
		\label{RE_iter_512}
	\end{figure}
	
	\begin{figure}[htbp]
		\centering
		\includegraphics[width=0.9\textwidth]{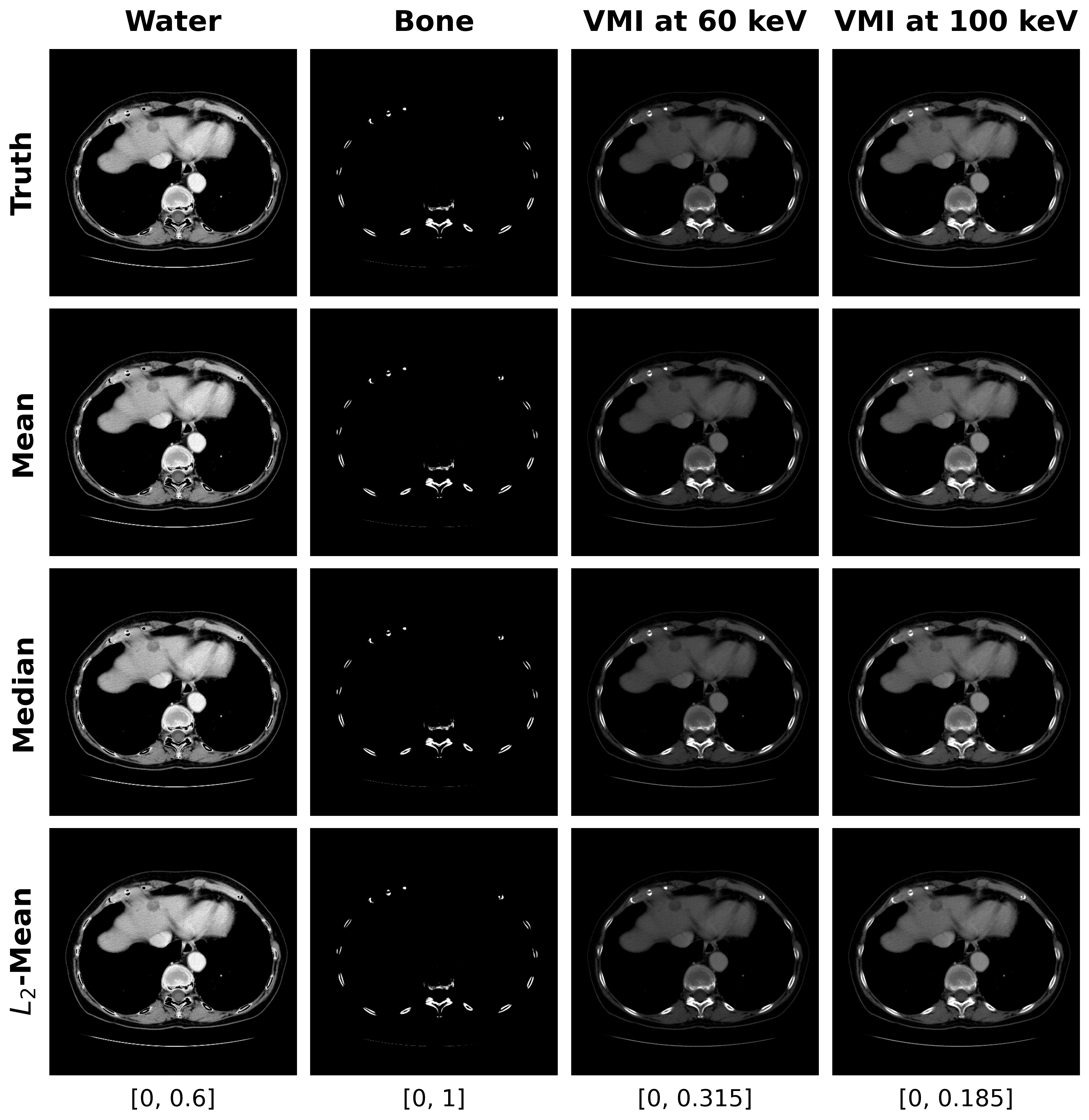}
		\caption{Reconstruction results in the noiseless case. From left to right: basis images of water and bone, VMIs at energies 60 keV and 100 keV of a clinical image. From top to bottom: the true images, the results after 100 iterations using the proposed algorithm with different aggregated energy spectra.}
		\label{result3}
	\end{figure}

\section{Discussion and conclusion}\label{section6}

	In this work, we proposed an efficient and accurate image reconstruction algorithm for geometric-inconsistent MSCT with ray-dependent energy spectra.
	
	Inspired by the AFIRE algorithm, we observed that the Jacobian matrix possesses a special structure that enables efficient inversion when the involved weight coefficients are independent of the ray indeces. Motivated by this insight, we introduced a ray-decoupling strategy for the weight coefficients. Specifically, we selected special points to decouple the basis sinograms from individual rays, and employed the aggregated energy spectra to decouple the energy spectra from the ray indeces. With this ray-wise decoupling, we derived an approximate Jacobian matrix that can be expressed as a block product of a diagonal matrix of projection matrices and a very small-scale matrix. Based on this approximation, we developed a new efficient and accurate reconstruction algorithm.
	
Under appropriate conditions, we further established the convergence theory of the proposed algorithm. Specifically, we proved that when the spectral discrepancy between the ray-dependent energy spectra and their aggregated version is sufficiently small, and when the residual at each iteration satisfies a certain contraction condition, the sequence generated by the algorithm converges to the unique solution of the nonlinear system.
	
We conducted both noiseless and noisy numerical experiments to demonstrate that the proposed algorithm can efficiently and accurately reconstruct basis images under geometric-inconsistent scanning configurations with ray-dependent energy spectra. Numerical comparisons revealed that the proposed method significantly outperforms the existing NCPD and NKM algorithms in terms of both reconstruction accuracy and computational efficiency. In the noiseless setting, the proposed algorithm achieved stable convergence within 60 iterations for $128 \times 128$ images and attained machine precision in just over ten seconds, whereas the other compared methods were far from convergence. In the noisy scenario, it converged within 10 iterations for $256 \times 256$ images and exhibited strong robustness to noise, whereas the compared methods failed to converge within 50 iterations. Furthermore, we investigated the impact of different choices for the aggregated energy spectra, namely, the mean, median, and $L_2$-mean aggregations, and found that the convergence, accuracy, and efficiency of the proposed algorithm remain robust against the specific choice of the aggregated energy spectra. These results confirmed the effectiveness, robustness, and practicality of the proposed algorithm for geometric-inconsistent MSCT reconstruction with ray-dependent energy spectra. 
	
Furthermore, the proposed algorithm demonstrates remarkable flexibility and wide applicability. It seamlessly handles both ray-dependent and ray-independent energy spectra, and is robust to scanning geometric parameters, accommodating both consistent and inconsistent configurations. A key advantage is its suitability for cases where the number of energy bins meets or exceeds the number of basis materials, making it particularly advantageous for photon-counting detector CT. Moreover, the framework can be readily extended to three-dimensional spectral CT imaging applications.
	
Future work will focus on uncertainty quantification of interest. In this work, both the proposed algorithm and numerical experiments assume that the ray-dependent energy spectra are precisely known. In practice, however, it is impossible to obtain exact knowledge of the energy spectra, as spectral estimation is an extremely ill-posed problem \cite{Pan_2019,Leinweber, 10081437,Zhao_estimate}. Therefore, we plan to investigate how spectral errors affect the proposed model and algorithm, and further develop uncertainty quantification methods to assess the reliability of the reconstruction results under spectral inaccuracies.

\bibliographystyle{plain}
\bibliography{reference}

\end{document}